\documentclass[pra,aps,showpacs,twocolumn,superscriptaddress]{revtex4}

\usepackage{amsmath,amsfonts,amssymb,caption,hyperref,color,epsfig,graphics,subfigure,graphicx,latexsym,mathrsfs,revsymb,theorem,url,verbatim,epstopdf}
\allowdisplaybreaks[3]

\hypersetup{colorlinks,linkcolor={blue},citecolor={blue},urlcolor={red}}

\usepackage{hyperref}

\makeatletter

\newcommand{\Rmnum}[1]{\expandafter\@slowromancap\romannumeral #1@}
\makeatother
\newtheorem{definition}{Definition}
\newtheorem{proposition}[definition]{Proposition}
\newtheorem{Lemma}[definition]{Lemma}

\newtheorem{Theorem}[definition]{Theorem}

\newtheorem{conjecture}[definition]{Conjecture}

\newtheorem{remark}[definition]{Remark}

\newtheorem{example}{Example}
\newtheorem{question}[definition]{Question}

\def\squareforqed{\hbox{\rlap{$\sqcap$}$\sqcup$}}
\def\qed{\ifmmode\squareforqed\else{\unskip\nobreak\hfil
		\penalty50\hskip1em\null\nobreak\hfil\squareforqed
		\parfillskip=0pt\finalhyphendemerits=0\endgraf}\fi}
\def\endenv{\ifmmode\;\else{\unskip\nobreak\hfil
		\penalty50\hskip1em\null\nobreak\hfil\;
		\parfillskip=0pt\finalhyphendemerits=0\endgraf}\fi}
\newenvironment{proof}{\noindent \textbf{{Proof.~} }}{\qed}
\def\Dbar{\leavevmode\lower.6ex\hbox to 0pt
	{\hskip-.23ex\accent"16\hss}D}
\makeatletter
\def\url@leostyle{%
	\@ifundefined{selectfont}{\def\UrlFont{\sf}}{\def\UrlFont{\small\ttfamily}}}
\makeatother
\def\bcj{\begin{conjecture}}
	\def\ecj{\end{conjecture}}
\def\bcr{\begin{corollary}}
	\def\ecr{\end{corollary}}
\def\bd{\begin{definition}}
	\def\ed{\end{definition}}
\def\bea{\begin{eqnarray}}
\def\eea{\end{eqnarray}}
\def\bem{\begin{enumerate}}
	\def\eem{\end{enumerate}}
\def\bex{\begin{example}}
	\def\eex{\end{example}}
\def\bim{\begin{itemize}}
	\def\eim{\end{itemize}}
\def\bl{\begin{lemma}}
	\def\el{\end{lemma}}
\def\bma{\begin{bmatrix}}
	\def\ema{\end{bmatrix}}
\def\bpf{\begin{proof}}
	\def\epf{\end{proof}}
\def\bpp{\begin{proposition}}
	\def\epp{\end{proposition}}
\def\bqu{\begin{question}}
	\def\equ{\end{question}}
\def\br{\begin{remark}}
	\def\er{\end{remark}}
\def\bt{\begin{theorem}}
	\def\et{\end{theorem}}

\def\btb{\begin{tabular}}
	\def\etb{\end{tabular}}

\newcommand{\nc}{\newcommand}

\nc{\bbA}{\mathbb{A}} \nc{\bbB}{\mathbb{B}} \nc{\bbC}{\mathbb{C}}
\nc{\bbD}{\mathbb{D}} \nc{\bbE}{\mathbb{E}} \nc{\bbF}{\mathbb{F}}
\nc{\bbG}{\mathbb{G}} \nc{\bbH}{\mathbb{H}} \nc{\bbI}{\mathbb{I}}
\nc{\bbJ}{\mathbb{J}} \nc{\bbK}{\mathbb{K}} \nc{\bbL}{\mathbb{L}}
\nc{\bbM}{\mathbb{M}} \nc{\bbN}{\mathbb{N}} \nc{\bbO}{\mathbb{O}}
\nc{\bbP}{\mathbb{P}} \nc{\bbQ}{\mathbb{Q}} \nc{\bbR}{\mathbb{R}}
\nc{\bbS}{\mathbb{S}} \nc{\bbT}{\mathbb{T}} \nc{\bbU}{\mathbb{U}}
\nc{\bbV}{\mathbb{V}} \nc{\bbW}{\mathbb{W}} \nc{\bbX}{\mathbb{X}}
\nc{\bbZ}{\mathbb{Z}}

\nc{\bA}{{\bf A}} \nc{\bB}{{\bf B}} \nc{\bC}{{\bf C}}
\nc{\bD}{{\bf D}} \nc{\bE}{{\bf E}} \nc{\bF}{{\bf F}}
\nc{\bG}{{\bf G}} \nc{\bH}{{\bf H}} \nc{\bI}{{\bf I}}
\nc{\bJ}{{\bf J}} \nc{\bK}{{\bf K}} \nc{\bL}{{\bf L}}
\nc{\bM}{{\bf M}} \nc{\bN}{{\bf N}} \nc{\bO}{{\bf O}}
\nc{\bP}{{\bf P}} \nc{\bQ}{{\bf Q}} \nc{\bR}{{\bf R}}
\nc{\bS}{{\bf S}} \nc{\bT}{{\bf T}} \nc{\bU}{{\bf U}}
\nc{\bV}{{\bf V}} \nc{\bW}{{\bf W}} \nc{\bX}{{\bf X}}
\nc{\bZ}{{\bf Z}}

\nc{\cA}{{\cal A}} \nc{\cB}{{\cal B}} \nc{\cC}{{\cal C}}
\nc{\cD}{{\cal D}} \nc{\cE}{{\cal E}} \nc{\cF}{{\cal F}}
\nc{\cG}{{\cal G}} \nc{\cH}{{\cal H}} \nc{\cI}{{\cal I}}
\nc{\cJ}{{\cal J}} \nc{\cK}{{\cal K}} \nc{\cL}{{\cal L}}
\nc{\cM}{{\cal M}} \nc{\cN}{{\cal N}} \nc{\cO}{{\cal O}}
\nc{\cP}{{\cal P}} \nc{\cQ}{{\cal Q}} \nc{\cR}{{\cal R}}
\nc{\cS}{{\cal S}} \nc{\cT}{{\cal T}} \nc{\cU}{{\cal U}}
\nc{\cV}{{\cal V}} \nc{\cW}{{\cal W}} \nc{\cX}{{\cal X}}
\nc{\cZ}{{\cal Z}}

\nc{\hA}{{\hat{A}}} \nc{\hB}{{\hat{B}}} \nc{\hC}{{\hat{C}}}
\nc{\hD}{{\hat{D}}} \nc{\hE}{{\hat{E}}} \nc{\hF}{{\hat{F}}}
\nc{\hG}{{\hat{G}}} \nc{\hH}{{\hat{H}}} \nc{\hI}{{\hat{I}}}
\nc{\hJ}{{\hat{J}}} \nc{\hK}{{\hat{K}}} \nc{\hL}{{\hat{L}}}
\nc{\hM}{{\hat{M}}} \nc{\hN}{{\hat{N}}} \nc{\hO}{{\hat{O}}}
\nc{\hP}{{\hat{P}}} \nc{\hR}{{\hat{R}}} \nc{\hS}{{\hat{S}}}
\nc{\hT}{{\hat{T}}} \nc{\hU}{{\hat{U}}} \nc{\hV}{{\hat{V}}}
\nc{\hW}{{\hat{W}}} \nc{\hX}{{\hat{X}}} \nc{\hZ}{{\hat{Z}}}

\nc{\hn}{{\hat{n}}}

\def\dim{\mathop{\rm Dim}}

\def\min{\mathop{\rm min}}

\def\tr{\mathrm{ tr}}

\newcommand{\bra}[1]{\langle#1|}
\newcommand{\ket}[1]{|#1\rangle}

\newcommand{\norm}[1]{\lVert#1\rVert}

\def\Dbar{\leavevmode\lower.6ex\hbox to 0pt
	{\hskip-.23ex\accent"16\hss}D}

\begin{document}
	\title{The Deimginarity Cost of Quantum States}
	
	\author{Xian Shi}\email[]
	{shixian01@gmail.com}
	\affiliation{College of Information Science and Technology,
		Beijing University of Chemical Technology, Beijing 100029, China}

	%
	
	
	
	\date{\today}
	\begin{abstract}
	Here we address a task denoted as deimaginarity, which is to transform a state into a real state with the aid of random covariant-free unitary operations. We consider the minimum cost of randomness required for deimaginarity in the scenario of infinite copies and sufficiently small error, and we prove that the deimaginarity cost of a state $\rho$ is equal to its regularized relative entropy of imaginarity, which can be seen as an operational interpretation of the relative entropy of imaginarity. 
	\end{abstract}

	\pacs{03.65.Ud, 03.67.Mn}
	\maketitle

\section{Introduction}
Quantum resource theories provide a fancy view to address the properties of quantum systems and present their applications in various quantum information tasks \cite{gour2015resource,chitambar2019quantum,gour2024resources}. Among them, entanglement is one of the most fundamental resources \cite{horodecki2009quantum,friis2019entanglement,erhard2020advances}. By consuming entanglement, varieties of protocols can be implemented that are impossible in the view of classical information theory, such as dense coding \cite{bennett1992communication}, and quantum telecommunication \cite{bennett1993teleporting}. Recently, many other quantum resources have been developed,such as coherence \cite{streltsov2017colloquium,hu2018quantum}, asymmetry \cite{marvian2012symmetry}, magic \cite{howard2017application,veitch2014resource}, nonlocality \cite{brunner2014bell}, steering \cite{cavalcanti2016quantum,uola2020quantum} and imaginarity \cite{hickey2018quantifying}.

Due to the postulates of quantum mechanics \cite{Nielsen_Chuang_2010}, quantum information theory builds on the complex field. Many significant results showed the necessity and values of the imaginary part \cite{stueckelberg1960quantum,wootters2014rebit,renou2021quantum,chiribella2023positive}. Nevertheless, until recent years, imaginarity has been studied in the view of resource theory \cite{hickey2018quantifying,wu2021resource,wu2021operational,kondra2023real,wu2024resource,zhang2024broadcasting}.  A fundamental problem of resource theory is how to quantify the resourcefulness of a quantum state \cite{hickey2018quantifying,kondra2023real,xue2021quantification}. Among the commonly used quantifiers in the imaginarity resource, the method to quantify the imaginarity in terms of the relative entropy, the relative entropy of imaginarity, is considered in \cite{xue2021quantification}. It is a strong imaginarity monotone and additive \cite{xue2021quantification}. However, as far as I know, the operational interpretation of the relative entropy of imaginarity (REI) is unknown.

In this paper, we address the above problem by analyzing a task which is to transform a state into a real state with the aid of random free unitary operators. We consider the above task in the scenarios of infinite copies and vanishing errors, and present the minimum cost of randomness per copy to deimaginarity is equal to the REI. This task has been considered in the resource theory of multipartite correlations \cite{groisman2005quantum}, entanglement \cite{berta2018disentanglement}, coherence \cite{singh2015erasing}, asymmetry \cite{wakakuwa2017symmetrizing}, non-markovianity \cite{9214470}, and dynamical resources \cite{liu2019resource}. 

This paper is organized as follows. In Sec. \ref{2}, we review the definitions, properties, and measures of quantums state in terms of the imaginarity resource. In Sec. \ref{3}, we present the definition of the deimaginarity cost of a state and the main result of this article. In Sec. \ref{4}, we end with a conclusion.

\section{Preliminary Knowledge}\label{2}
\indent In this section, we recall the definitions of the imaginary of quantum states. We also review some quantifiers of the imaginary resource of quantum states.
\subsection{The Resource Theory of Imaginarity}

Here we consider a quantum system described by $\mathcal{H}$ with finite dimensions. We denote $\mathcal{L}(\mathcal{H})$ as the set of linear operators of $\mathcal{H}$ and $\mathcal{D}(\mathcal{H})$ as the set of states on $\mathcal{H}.$ Before reviewing the resource theory of imaginarity \cite{hickey2018quantifying}, we fix an orthonormal reference basis $\{\ket{i}\}_{i=0}^{d-1}$ of $\mathcal{H}$ with its dimension $d$.

In the imaginarity theory, the real states are called the free states, which is the set $\mathcal{S}_R$ of a quantum state $\rho$ with a real density matrix, 
\begin{align}
\rho=\sum_{j,k}\rho_{jk}\ket{j}\bra{k},
\end{align}
here $\rho_{jk}\in\mathbb{R}.$ In this manuscript, we denote the set of all real density matrices as $\mathcal{F}.$ As all density matrices are Hermitian, a state $\rho\in \mathcal{F}$ if and only if $\rho^T=\rho.$ Due to the definitions of $\mathcal{F},$  $\mathcal{F}$ is convex and affine. Besides, the set $\mathcal{F}$ depends on the orthonormal reference basis $\{\ket{i}\}_{i=0}^{d-1}$.

Next we present the formulation of the free operations of this resource theory. In general, quantum channels can be represented by a set of Kraus operations $\{K_i\}$ with $\sum_iK_i^{\dagger}K_i=\mathbb{I}$. The free operations of imaginarity theory can be specified by the Kraus operations $\{K_i\}$ with $\bra{j}K_i\ket{k}\in\mathbb{R}$, here $\ket{j}$ and $\ket{k}$ takes over all the elements in the basis $\{\ket{i}\}_{i=0}^{d-1}$. Due to the properties of free operations of imaginarity theory, the real operations cannot creat imginarity from real states.

Assume $U$ is a unitary operator acting on the system $\mathcal{H}$, $\Lambda(\cdot)=U\cdot U^{\dagger}$ is RNG if and only $\Lambda(\cdot)=O\cdot O^T$, here $O$ is a real orthogonal operator. Let $\Theta(\cdot)=\frac{(\cdot)^T+(\cdot)}{2}$. When $\Lambda(\cdot)$ satisfies $\Lambda\circ\Theta=\Theta\circ\Lambda$, $\Lambda(\cdot)$ is covariant. Next if $\Omega(\cdot)=O(\cdot)O^T$ is covariant, and $\sigma$ is a free state 
\begin{align*}
\Omega(\sigma)=&\frac{1}{2}[\Omega(\sigma^T)+\Omega(\sigma)]\\
=&\Omega\circ\Theta(\sigma)\\
=&\Theta\circ\Omega(\sigma)\\
=&\frac{1}{2}[\Omega(\sigma)^T+\Omega(\sigma)],
\end{align*}
that is, $\Omega(\sigma)=\Omega(\sigma)^T$. Hence, $\Omega(\sigma)$ is a free state. That is, a covariant-free unitary operation cannot transform a free state into a resourceful state. 

Before we introduce the process of deimaginarity of a generic state, we present the process of deimaginarity of the maximally resourceful state, $\ket{\phi}=\frac{1}{\sqrt{2}}(\ket{0}+i\ket{1})$. The deimaginary process can be achieved by two unitary operators $\mathbb{I}$ and $\sigma_z$ with equal probability,
\begin{align*}
	\ket{\phi}\bra{\phi}\longrightarrow \frac{1}{2}(\ket{\phi}\bra{\phi}+\sigma_z\ket{\phi}\bra{\phi}\sigma_z)=\frac{1}{2}\mathbb{I}_2.
\end{align*}
Note that $\mathbb{I}_2\in\mathcal{F}$, and $\Theta(\sigma_z\cdot\sigma_z)=\sigma_z\Theta(\cdot)\sigma_z.$ That is, by the application of two covariant unitary operations with equal probability, we could erase the imaginarity of the maximally resourceful state. The above effects can also be implemented for a $d$ dimensional quantum system. Let $\mathcal{H}_d$ is a $d$-dimensional system, any state $\rho$ acting on $\mathcal{H}_d$ can be decomposed as 
\begin{align*}
\rho=Re(\rho)+iIm(\rho),
\end{align*}
here $Re(\rho)=\frac{\rho+\rho^T}{2}$ is a real symmetric quantum state, and $Im(\rho)=\frac{1}{2i}(\rho-\rho^T)$ is a real anti-symmetric state. Moreover, there exists a real $O$ \cite{horn} such that
\begin{align*}
	OIm(\rho)O^T=\mathrm{0}_{d-2r}\bigoplus_{k=1}^r\lambda_k\begin{pmatrix}
		0&1\\
		-1&0
	\end{pmatrix},
\end{align*}
where $\lambda_k>0.$ For a real orthogonal matrix $O$, we have 
\begin{align}
	\Theta(O\cdot O^T)=O\Theta(\cdot)O^T,
\end{align}
Besides, for a state $\rho$ of two-dimensional systems, we have
\begin{align}
	\frac{1}{4}\sum_{j,k=1}^{2}\sigma_x^k\sigma_z^j\rho{\sigma_z^j}^{\dagger}{\sigma_x^k}^{\dagger}=\frac{1}{2}\mathbb{I}_2.
\end{align}
That is, there always exists an ensemble of covariant-free unitaries to erase the imaginarity of a quantum state.

Next we present the resource theory of imaginarity of the composite systems. Assume $\mathcal{H}$ is a Hilbert space with reference orthogonal basis $\{\ket{i}\}$, let $\mathcal{H}^{\otimes n}$ be a system composed of $n$ duplicates of $\mathcal{H}$ with reference orthogonal basis $\{\ket{i_1i_2\cdots i_n}\}$. A state $\rho$ acting on $\mathcal{H}^{\otimes n}$ is free in terms of the imaginarity resource if and only if $\rho^T=\rho$, here $(\cdot)^T$ is with respect to the basis $\ket{i_1i_2\cdots i_n}$. As $\mathcal{H}^{\otimes n}$ is a composite space composed of $n$ duplicates of $\mathcal{H}$, the transpose operator $(\cdot)^T$ of $\mathcal{H}^{\otimes n}$ can be seen as $(\cdot)^T=\otimes_k(\cdot)^{T_k}$, here $(\cdot)^{T_k}$ is the transpose operator of $\mathcal{H}_k$ with repect to the reference basis $\{\ket{i}_k\}.$ 

\subsection{Relative Entropy of the Resource Theory of Imaginarity}
\indent To quantify the imaginarity of a quantum state, Xue $et$ $al.$ considered a quantifier called the relative entropy of imaginarity (REI) $\mathcal{I}_r(\cdot)$ of a quantum state \cite{xue2021quantification},
\begin{align}
\mathcal{I}_r(\rho)=\min_{\sigma\in\mathcal{R}}D(\rho||\sigma),\label{f11}
\end{align}
where the minimum takes over all the real states $\sigma$, and $D(\rho||\sigma)=\tr[\rho(\log\rho-\log\sigma)].$ Moreover, the REI $\mathcal{I}_r(\rho)$ can be written as 
\begin{align}
\mathcal{I}_r(\rho)=S(Re(\rho))-S(\rho),\label{f12}
\end{align} 
where $Re(\rho)=\frac{\rho+\rho^T}{2}$, and $S(\rho)=-\tr(\rho\log\rho)$. The REI is additive.

Up to now, there exists few results on the REI, $\mathcal{I}_r(\cdot)$, here we present a direct operational meaning of the REI.

\section{Main Results}\label{3}
Assume $\rho$ is a state of a system with finite dimensions. A basic task is to transform a state $\rho$ to a real state by applying covariant-free unitary operations randomly. A fundamental problem is to minimize the cost of randomness to make the imaginary state free.

\begin{figure}
		\centering
	\includegraphics[width=90mm]{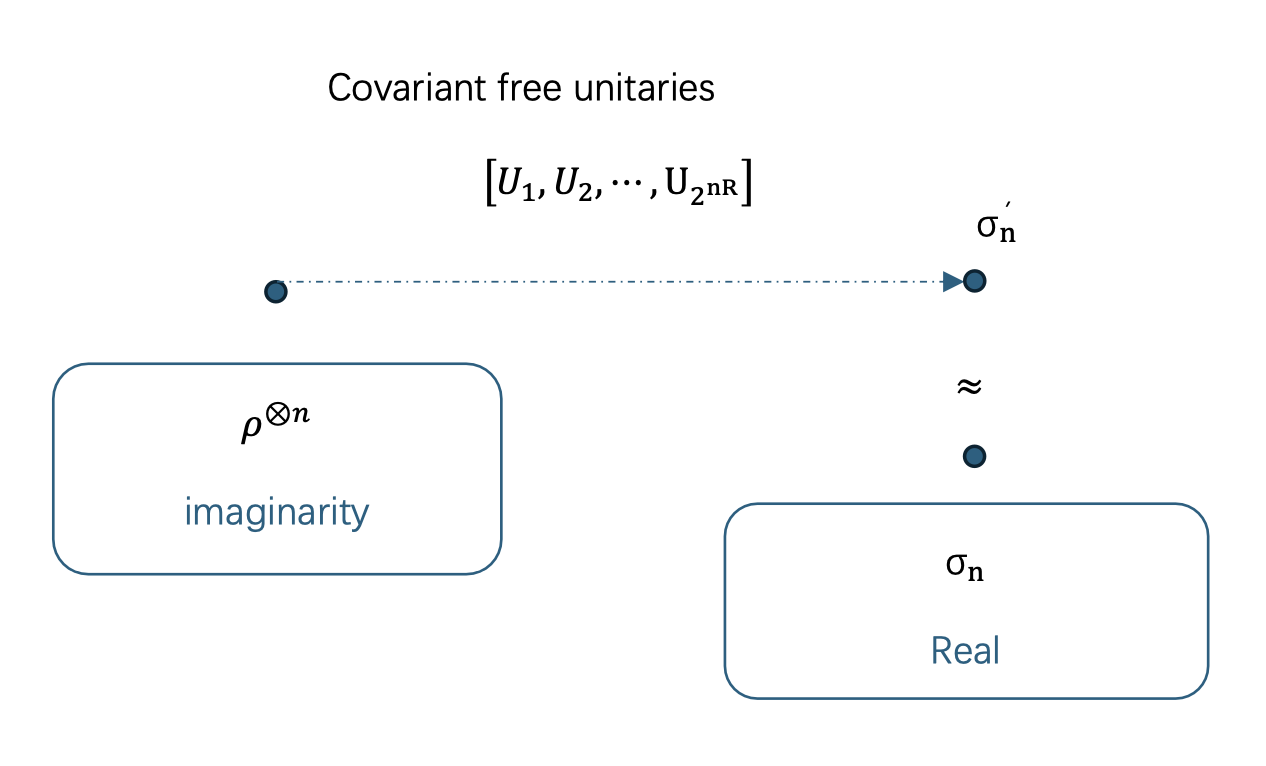}\\
	
	\captionsetup{justification=raggedright,singlelinecheck=false}
	
	\caption{The task of deimginarity. Here $n$-copies of a state $\rho$ is transformed by a set of covariant-free unitaries $\{U_1,U_2,\cdots,U_{2^{nR}}\}$. This task requires the final state $\sigma_n^{'}$ is real with respect to the reference orthonormal basis up to a small $\epsilon>0$ in terms of the $1$-norm. The deimaginarity cost of a state $\rho$ with respect to the reference basis is defined as the minimum $R$ such that the above task is completed by properly choosing $2^{nR}$ covariant-free unitaries and free states $\sigma_n$ when $n\rightarrow\infty.$}\label{fig1}
\end{figure}

\begin{definition}
	Assume $\rho$ is a state of $\mathcal{H}$ with finite dimensions. A rate $R$ is denoted to be achievable to make $\rho$ free with respect to the reference basis $\{\ket{i}\}_{i=0}^{d-1}$ of $\mathcal{H}$ if, $\forall\epsilon>0,$ there exists a large $n$ such that there exists a free state $\sigma_n$ and a set of real ortthogonal operators $\{O_k\}_{k=1}^{2^{nR}}$ with $O_k\Theta(\cdot)O_k^{\dagger}=\Theta(O_k\cdot O_k^{\dagger})$ such that 
	\begin{align*}
	\norm{\mathcal{O}_n(\rho^{\otimes n})-\sigma_n}_1\le\epsilon,
	\end{align*}
	here 
	\begin{align*}
	\mathcal{O}_n(\gamma)=\frac{1}{{2^{nR}}}\sum_{k=1}^{2^{nR}}O_k\gamma O_k^{\dagger}.
	\end{align*}
	
	The deimaginarity cost of a state $\rho$ with respect to the reference basis $\{\ket{i}\}_{i=0}^{d-1}$ is 

	\begin{align*}
	&C_d(\rho)=\\
	&\inf\{R|\textit{$R$ is achievable to make $\rho$ free w. r. t. $\{\ket{i}\}_{i=0}^{d-1}$}\}.
	\end{align*}

\end{definition}

Here the schematic diagram of the task of deimaginarity cost of a state is plotted in Fig \ref{fig1}.

The main result of this article is to present the analytical representation of the deimaginarity cost of a state $\rho$, $C_d(\rho)$. 

\begin{Theorem}\label{t1}
	Assume $\rho$ is a state acting on the Hilbert space $\mathcal{H},$ then
	\begin{align*}
	C_d(\rho)=\mathcal{I}_r^{\infty}(\rho)=:\lim_{n\rightarrow \infty}\frac{1}{n}\mathcal{I}_r(\rho^{\otimes n})
	\end{align*}
\end{Theorem}

The proof of this theorem is placed in Sec. \ref{app}
\section{Conclusion}\label{4}
In this paper, we considered the task of deimaginarity, and analyzed the minimal cost of randomness needed to make a state real. Particularly, we addressed the transformation task in the scenarios of an asymptotic limit of infinite copies and vanishing error, and we showed the minimum cost of the randomness for a state to make it real is equal to the regularized relative entropy of imaginarity of the state. At last, it would be meaningful to study the task in terms of other resources of some quantum system, such as steering \cite{cavalcanti2016quantum,uola2020quantum} and contextuality \cite{budroni2022kochen,pavicic2023quantum}.
  \section{Acknowledgement}
X. S. was supported by the National Natural Science Foundation of China (Grant No. 12301580).
\bibliographystyle{IEEEtran}
\bibliography{ref}

\setcounter{equation}{0}

\renewcommand\theequation{A.\arabic{equation}}
\section{Appendix}\label{app}
\subsection{Proof of Theorem \ref{t1}}
Here we prove Theorem \ref{t1} separablely, first we prove the converse part $C_d(\rho)\ge \mathcal{I}^{\infty}_r(\rho)$ first, then we show $C_d(\rho)\le \mathcal{I}^{\infty}_r(\rho)$. Combing the proof of the both parts, we finish the proof of Theorem \ref{t1}.
\subsubsection{Converse Part}
\indent Assume $R\ge C_d(\rho)$, for any $\epsilon>0$ and sufficiently large $n$, there exists a free state $\sigma_n\in \mathcal{F}(\mathcal{H}^{\otimes n})$ and a set of unitaries $\{O_k\}_{k=1}^{2^{nR}}$ such that 
\begin{align}
\norm{\frac{1}{2^{nR}}\sum_{k=1}^{2^{nR}}O_k\rho O_k^{T}-\sigma_n}_1\le \epsilon,
\end{align}

here $O_k$ satisfies $\Theta(O_k\cdot O_k^T)=O_k\Theta_n(\cdot)O_k^T.$ Next let $\mathcal{O}_n(\cdot)=\frac{1}{2^{nR}}\sum_{k=1}^{2^{nR}}O_k\cdot  O_k^{T}$ ,as $\sigma_n\in \mathcal{F}(\mathcal{H}_n)$, $\Theta(\sigma_n)=\sigma_n$,
\begin{align}
&\norm{(\Theta\circ\mathcal{O}_n)(\rho^{\otimes n})-\mathcal{O}_n(\rho^{\otimes n})}_1\\
\le&\norm{(\Theta\circ\mathcal{O}_n)(\rho^{\otimes n})-\sigma_n}_1+\norm{\mathcal{O}_n(\rho^{\otimes n})-\sigma_n}_1\\
=& \norm{(\Theta\circ\mathcal{O}_n)(\rho^{\otimes n})-\Theta(\sigma_n)}_1+\norm{\mathcal{O}_n(\rho^{\otimes n})-\sigma_n}_1\\
\le& 2\norm{\mathcal{O}_n(\rho^{\otimes n})-\sigma_n}_1\le 2\epsilon,\label{l1}
\end{align}

here the first inequality is due to the triangle inequality of the $1$-norm, the equality is due to that $\sigma_n$ is free, the second inequality is due to the contractive property of the trace-preserving quantum operations under the $1$-norm.

Assume $E$ is an ancillary system with dimension $2^{nR}$, and $\mathcal{T}:\mathcal{H}_A^{\otimes n}\rightarrow \mathcal{H}_E\otimes\mathcal{H}_A^{\otimes n}$,
\begin{align*}
\mathcal{T}=\frac{1}{2^{nR}}\sum_{k=1}^{2^{nR}}\ket{k}_E\otimes O_k^{A^n},
\end{align*}

Next let $\ket{\psi}_{AZ}$ is a purification of $\rho,$ let
\begin{align*}
\ket{\overline{\psi}_n}_{EA^nZ^n}=\mathcal{T}\ket{\psi^{\otimes n}}_{A^nZ^n},
\end{align*}

then
\begin{align}
nR=&\log\dim\mathcal{H}_E\nonumber\\
\ge& S(E)_{\overline{\psi}_n}\nonumber\\
\ge& S(EZ_n)_{\overline{\psi}_n}-S(Z_n)_{\overline{\psi}_n}\nonumber\\
=&S(\mathcal{O}_n(\rho^{\otimes n}))-nS(\rho),\label{l2}
\end{align}

here the second inequality is due to the subadditivity of von Neumann entropy $S(\cdot)$, the second equality follows from
\begin{align*}
S(EZ_n)_{\overline{{\psi}_n}}=&S(A_n)_{\overline{\psi}_n}=S(\mathcal{O}_n(\rho^{\otimes n})),\\
S(Z_n)_{\overline{{\psi}_n}}=&S(Z^n)_{\psi^{\otimes n}}=S(\rho^{\otimes n})=nS(\rho).
\end{align*}

Next due to the inequality (\ref{l1}) and the Fannes inequality \cite{fannes1973continuity}, we have
\begin{align}
S(\mathcal{O}_n(\rho^{\otimes n}))\ge S((\Theta\circ\mathcal{O}_n)(\rho^{\otimes n}))-n\eta(2\epsilon)\log d,\label{l21}
\end{align}

here $\eta(\cdot)$ satisfies $\lim_{x\rightarrow 0}\eta(x)=0.$  Then due to the concavity of the von Neumann entropy, we have
\begin{align}
&S(\mathcal{O}_n(\rho^{\otimes n}))\\
\ge& \frac{1}{2^{nR}}\sum_{k=1}^{2^{nR}}S(\Theta(O_k\rho^{\otimes n}O_k^{T}))\\
\ge&S(\Theta(\rho^{\otimes n})),\label{l3}
\end{align}
here the first inequality is due to the concavity of $S(\cdot)$, the second inequality is due to that $O_k$ is covariant-free and Lemma \ref{l6}.

Combing (\ref{l2}), (\ref{l21}), and (\ref{l3}), we have
\begin{align}
nR\ge S(\Theta(\rho^{\otimes n}))-nS(\rho)-n\eta(2\epsilon)\log d,
\end{align}
that is,
\begin{align}
R\ge \lim_{n\rightarrow \infty}\frac{1}{n}\mathcal{I}_r(\rho^{\otimes n})-\eta(2\epsilon)\log d.
\end{align}

As the above inequality holds for all $R\ge C_d(\rho)$ and arbitrarily small $\epsilon>0$, we finish the proof of the converse part.

\subsubsection{Direct Part}
\indent In this subsection, we will present the proof of the following inequality,
\begin{align*}
C_d(\rho)\le \lim_{n\rightarrow \infty}\frac{1}{n}\mathcal{I}_r(\rho^{\otimes n}).
\end{align*}
The method of this proof is similar to the one shown in \cite{groisman2005quantum,wakakuwa2017symmetrizing}. The preliminary knowledge needed here are placed in Sec. \ref{appm}.

Assume $\mathcal{H}_{n,\delta}$ is the $\delta$-typical subspacces with respect to $\rho$. $\Pi_{n,\delta}$ is the projection onto the space $\mathcal{H}_{n,\delta}$. Due to the properties of $\delta$-typical subspace presented in Sec. \ref{appm}, for any $\epsilon\ge0,$ there exists a large $n$ such that 
\begin{align}
\tr[\Pi_{n,\delta}\rho^{\otimes n}\Pi_{n,\delta}]\ge 1-\epsilon,\label{f1}\\
2^{n(S(\rho)-\delta)}\Pi_{n,\delta}\rho^{\otimes n}\Pi_{n,\delta}\le\Pi_{n,\delta}.\label{f2}
\end{align}

Based on $(\ref{f1})$ and the Gentle Measurement Lemma \ref{l5}, we have
\begin{align}
	\norm{\rho^{\otimes n}-\Pi_{n,\delta}\rho^{\otimes n}\Pi_{n,\delta}}_1\le 2\sqrt{2}\epsilon.\label{fg}
\end{align}

Here we denote ${D}_{n,\delta}$ as the dimension of $\mathcal{H}_{n,\delta}$, then 
\begin{align*}
 {D}_{n,\delta}\le& 2^{n(S(\rho)+\delta))}\\
 \le& 2^{S(\Theta(\rho^{\otimes n}))+n\delta}.
\end{align*}

The last inequality is due to the additivity of $S(\cdot)$ for product states and Lemma \ref{l4}.

Assume $\{p(dU),U\}$ is an ensemble of the covariant-free unitary operators such that for any state $\theta$ acting on the Hilbert space $\mathcal{H}_{n,\delta}$, $\int p(dU)U\theta U^{\dagger}=\frac{1}{D_{n,\delta}}\mathbb{I}_{n,\delta}$, here $\mathbb{I}_{n,\delta}$ is the identity acting on the space $\mathcal{H}_{n,\delta}$.
Hence, we have 
\begin{align}
\int_U p(dU)U\Pi_{n,\delta}\rho^{\otimes n}\Pi_{n,\delta}U^{\dagger}=\frac{\mu_{n,\delta}}{D_{n,\delta}}\mathbb{I}_{n,\delta},
\end{align}
here $\mu_{n,\delta}=\tr(\Pi_{n,\delta}\rho^{\otimes n}\Pi_{n,\delta}).$

 Next let us denote $W=U\Pi_{n,\delta}\rho^{\otimes n}\Pi_{n,\delta}U^{\dagger}$ as random operators with its distribution $p(dU)$. Let $D^{'}_{n,\delta}=2^{n(S(\rho)-\delta)}.$ Based on (\ref{f2}), we have 
 \begin{align}
 W^{*}:=D^{'}_{n,\delta}W=& 2^{n(S(\rho)-\delta)} U\Pi_{n,\delta}\rho^{\otimes n}\Pi_{n,\delta}U^{\dagger}\nonumber\\
 \le& U\Pi_{n,\delta}U^{\dagger},
 \end{align}
 hence $W^{*}\in[0,\mathbb{I}]$.
  Let 
 \begin{align}
 \mathbb{E}(W^{*})=&D_{n,\delta}^{'}\int_Up(dU)U\Pi_{n,\delta}\rho^{\otimes n}\Pi_{n,\delta}U^{\dagger}\nonumber,
 \end{align}
then we have the minimum nonzero eigenvalue $\lambda_{n,\delta}$ of $\mathbb{E}(W^{*})$ is bounded by
\begin{align}
\lambda_{n,\delta}\ge 2^{-(\mathcal{I}_r(\rho^{\otimes n})+2n\delta)}(1-\epsilon).
\end{align}

If $W_1,W_2,\cdots,W_N$ are $N$ \emph{i.i.d.} $W^{*}$, then by Lemma \ref{ocb}, we have
\begin{align*}
&P[(1-\epsilon)\mathbb{E}(W^{*})\le \overline{W}\le (1+\epsilon)\mathbb{E}(W^{*})]\\
\ge& 1-2d^ne^{-\frac{N\lambda_{n,\delta}\epsilon^2}{2}},
\end{align*}
Denote $N=2^{\mathcal{I}_r(\rho^{\otimes n})+3n\delta}$, the right hand side of the above inequality is nonzero when $n$ is sufficiently large. Then there exists $\frac{1}{N}\sum_i W_i$ such that 
\begin{align}
\norm{\frac{1}{N}\sum_iU_i\Pi_{n,\delta}\rho^{\otimes n}\Pi_{n,\delta}U_i^{\dagger}-E(W)}_1\le \epsilon,\label{ff}
\end{align} 
here $E(W)$ is real, hence, we have
\begin{align*}
&\norm{\frac{1}{N}\sum_iU_i\rho^{\otimes n}U_i^{\dagger}-\frac{1}{\mu_{n,\delta}}\mathbb{E}(W)}_1\\
\le&\norm{\frac{1}{N}\sum_iU_i\Pi_{n,\delta}\rho^{\otimes n}\Pi_{n,\delta}U_i^{\dagger}-\frac{1}{N}\sum_iU_i\rho^{\otimes n}U_i^{\dagger}}\nonumber\\
+&\norm{\frac{1}{N}\sum_iU_i\Pi_{n,\delta}\rho^{\otimes n}\Pi_{n,\delta}U_i^{\dagger}-\mathbb{E}(W)}_1+\frac{1-\mu_{n,\delta}}{\mu_{n,\delta}}\norm{\mathbb{E}(W)}_1\nonumber\\
\le&2\epsilon+2\sqrt{\epsilon},
\end{align*}
here
the first inequality is due to the triangle inequality of the $1$-norm, the second inequality is due to (\ref{fg}) and (\ref{ff}).
Note that $\epsilon$ and $\delta$ in the above proof can be chosen arbitrarily small, $R$ is achievable if $R\ge  \mathcal{I}_r^{\infty}(\rho)$,
hence, we finish the proof.

\subsection{Mathematical Tools}\label{appm}
In this subsection, we present necessary knowledge needed in the proof of the first two subsections of this appendix.

\subsubsection{Typical Sequences and Subspaces}
Next we recall the definitions and properties of typical subspaces \cite{thomas2006elements}. Assume $X$ is a discrete random variable with finite alphabet $\mathcal{X}$ and $p_x=P(X=x)$ when $x\in\mathcal{X}$. A sequence $\vec{x}=(x_1,x_2,\cdots,x_n)$ is denoted as $\delta$\emph{-typical with respect to }$\{p_x\}$ if it owns the following property,
\begin{align*}
2^{-n(H(X)+\delta)}\le\Pi_{i=1}^Np_i\le 2^{-n(H(X)+\delta)},
\end{align*}

where $H(X)=-\sum_xp_x\ln p_x$ is the Shannon entropy of $X$. The set of all $\delta$\emph{-typical sequences} is called the $\delta$\emph{-typical set}, which is denoted as $\mathcal{T}_{n,\delta}$. The set $\mathcal{T}_{n,\delta}$ owns the following property,
\begin{align*}
|\mathcal{T}_{n,\delta}|\le 2^{n(H(X)+\delta)}.
\end{align*}

Moreover, based on the wek law of large numbers, for arbitrary positive $\epsilon$, $\delta$, and sufficiently large $n$.

\begin{align*}
P((X_1,X_2,\cdots,X_n)\in\mathcal{T}_{n,\delta})\ge 1-\epsilon.
\end{align*}

Assume $\rho$ is a state acting on the system $\mathcal{H}$, $\rho=\sum_ip_i\ket{i}\bra{i}$ The $\delta$-typical subspace $\mathcal{H}_{n,\delta}\subset \mathcal{H}^{\otimes n}$ with respect to $\rho$ is 
\begin{align*}
\mathcal{H}_{n,\delta}=span\{\ket{i_1i_2\cdots i_n}\in\mathcal{H}^{\otimes n}|(x_1,x_2,\cdots,x_n)\in\mathcal{T}_{n,\delta}\},
\end{align*}

where $\mathcal{T}_{n,\delta}$ is the $\delta$-typical set with respect to $\{p_i\}$. Assume $\Pi_{n,\delta}$ is the projection of $\mathcal{H}_{n,\delta}$, then
\begin{align*}
2^{n(S(\rho)-\delta)}\Pi_{n,\delta}\rho^{\otimes n}\Pi_{n,\delta}\le \Pi_{n,\delta},\\
\dim\mathcal{H}_{n,\delta}\le 2^{n(S(\rho)+\delta)}.
\end{align*}
For any $\epsilon,\delta>0,$ and sufficiently large $n$, we have
\begin{align*}
\tr[\Pi_{n,\delta}\rho^{\otimes n}\Pi_{n,\delta}]=\sum_{\vec{x}\in\mathcal{T}_{n,\delta}}p_{\vec{x}}\ge1-\epsilon.
\end{align*}

\subsubsection{Necessary Lemmas}

\begin{Lemma}[Gentle Measurement Lemma]\cite{winter1999coding}\label{l5}
	Assume $\rho$ is a density matrix with $\tr\rho\le 1$ acting on a finite dimensional system, then for any $\epsilon\ge 0$ such that $0\le X\le \mathbb{I}$ and $\tr(\rho X)\ge 1-\epsilon$, we have
	\begin{align*}
		\norm{\rho-\sqrt{X}\rho\sqrt{X}}\le 2\sqrt{2}\epsilon.
	\end{align*}
\end{Lemma}

\begin{Lemma}[Fannes Inequality]\cite{fannes1973continuity}
	Assume $\rho$ and $\sigma$ are two states of a quantum system $\mathcal{H}$ with $\dim\mathcal{H}=d$, when $\norm{\rho-\sigma}_1\le \epsilon$, we have
	\begin{align*}
		|S(\rho)-S(\sigma)|\le \eta(\epsilon)\log d,
	\end{align*}
	where 
	\begin{equation*}
		\eta(x)= 
		\begin{cases}
			x-x\log x & \text{if } x \le \frac{1}{e} \\
			x+\frac{1}{e} & \text{if } x >\frac{1}{e}
		\end{cases}
	\end{equation*}
\end{Lemma}

\begin{Lemma}[Operator Chernoff Bound]\cite{ahlswede2002strong}\label{ocb}
	Assume $X_1,X_2,\cdots,X_N$ are \emph{i. i. d.} random variables taking values in the operator interval $[0,\mathbb{I}]$ with expectation $\mathbb{E}X_i=M\ge \mu \mathbb{I}$, then for $\epsilon\in[0,1]$, and denote $\overline{X}=\frac{1}{N}\sum_{i=1}^NX_i,$
	\begin{align*}
		Pr[\overline{X}\notin [(1-\epsilon)M,(1+\epsilon)M]]\le 2d e^{-\frac{N\mu\epsilon^2}{2}}
	\end{align*}
\end{Lemma}

\begin{Lemma}\label{l6}
	Assume $\rho\in \mathcal{D}(\mathcal{H})$, and $O$ is a covariant-free unitary operation, then
	\begin{align*}
		S(\Theta(O\rho O^{\dagger}))=S(\Theta(\rho)).
	\end{align*}
\end{Lemma}
\begin{proof}
	Due to (\ref{f11}) and (\ref{f12}), we have 
	\begin{align}
		S(\Theta(\rho))-S(\rho)=&\min_{\sigma\in \mathcal{R}}D(\rho||\sigma),\label{f14}\\
				S(\Theta(O\rho O^T))-S(O\rho O^T)=&\min_{\sigma^{}\in \mathcal{R}}D(O\rho O^T||\sigma),\label{f15}
	\end{align}
	as $S(\cdot)$ remains the same under the unitary operation, when $O$ is a real orthogonal matrix, 
	\begin{align}
		S(\rho)=S(O\rho O^T).\label{f13}
	\end{align}
	And
	\begin{align}
		\min_{\sigma\in\mathcal{R}}D(\rho||\sigma)=&\min_{\sigma^{'}\in \mathcal{R}}D(\rho||O^T\sigma^{'}O)\nonumber\\
		=&\min_{\sigma^{'}\in \mathcal{R}}D(O\rho O^T||\sigma^{'}),\label{f16}
	\end{align}
	here the last equality is due to that 
	\begin{align*}
		D(\rho||O^T\sigma^{'}O)=&\tr \rho\log\rho-\tr\rho\log O^T\sigma^{'}O\\
		=&\tr O\rho O^T\log O\rho O^T-\tr O\rho O^T\log \sigma^{'}\\
		=&D(O\rho O^T||\sigma^{'}).
	\end{align*}
	At last, combing (\ref{f14}),(\ref{f15})  and (\ref{f16}), we have 
	\begin{align*}
		S(\Theta(O\rho O^T))=S(\Theta(\rho)).
	\end{align*}
\end{proof}
\begin{Lemma}\label{l4}
	Assume $\rho$ is a state, then 
	\begin{align}
	S(Re(\rho)):=S(\frac{\rho+\rho^T}{2})\ge S(\rho).
	\end{align}
\end{Lemma}
\begin{proof}
	As $\rho\ge 0$, then each eigenvalue $\lambda$ of $\rho$ is nonnegative. Let $\lambda$ is arbitrary eigenvalue of $\rho$ with its corresponding eigenvector is $\ket{v},$
	\begin{align*}
	\rho^T\ket{\overline{v}}=&\overline{\rho}\ket{\overline{v}}\\
	=&\overline{\rho\ket{v}}\\
	=&\lambda\overline{\ket{v}},
	\end{align*}
	that is $\lambda$ is also the eigenvalue of $\rho^T$. Hence, $S(\rho^T)=S(\rho).$ Due to the concativity of $S(\cdot),$ we have
	\begin{align*}
	S(\frac{\rho+\rho^T}{2})\ge\frac{S(\rho^T)+S(\rho)}{2}=S(\rho).
	\end{align*}
\end{proof}
 \end{document}